\documentclass[a4paper,11pt]{article}

\usepackage{authblk}

\usepackage{amsmath}
\usepackage{amsthm}
\usepackage{amssymb}

\usepackage{algorithmic}
\usepackage{algorithm}

\usepackage{graphicx}

\usepackage[left=1.0in,top=1.0in,right=1.0in,bottom=1.0in,nohead]{geometry}

\theoremstyle{definition}
\newtheorem{definition}{Definition}

\newtheorem{theorem}{Theorem}
\newtheorem{corollary}[theorem]{Corollary}
\newtheorem{lemma}[theorem]{Lemma}

\newenvironment{ap_lemma}[1]{\par\noindent{\bf Lemma~#1.} \em}{}
\newenvironment{ap_theorem}[1]{\par\noindent{\bf Theorem~#1.} \em}{}
\newenvironment{ap_corollary}[1]{\par\noindent{\bf Corollary~#1.} \em}{}

\newtheorem{problem}{Problem}

\graphicspath{{figures/}}

\def\MC#1{{\mathcal #1}}
\newcommand{\BIGLR}[3]{{\left#1#3\right#2}}
\newcommand{\BIGP}[1]{{\BIGLR{(}{)}{#1}}}
\newcommand{\BIGC}[1]{{\BIGLR{|}{|}{#1}}}

\long\def\longdelete#1{}

\usepackage{color}

\def\REV#1{{#1}}

\title{Approximating Metrics by Tree Metrics of Small Distance-Weighted Average Stretch
}

\author[1]{Mong-Jen Kao\footnote{This work was done when the author was with Karlsruhe Institute of Technology (KIT), as a visiting scholar under the NSC-DAAD-sponsored sandwich program (grant number NSC99-2911-I-002-055-2).}}
\author[1]{D.T. Lee\footnote{Also with Academia Sinica, Taiwan. The author's present address is Department of Computer Science and Engineering, National Chung-Hsing University, Taichung, Taiwan.}}
\author[2]{Dorothea Wagner}

\affil[1]{Department of Computer Science and Information Engineering, \\
National Taiwan University, Taipei, Taiwan.\\
\texttt{d97021@csie.ntu.edu.tw, dtlee@nchu.edu.tw}}
\affil[2]{Faculty of Informatics, \\
Karlsruhe Institute of Technology (KIT), Germany.\\
\texttt{dorothea.wagner@kit.edu}}

\date{}

\begin{document}

\maketitle


\begin{abstract}
\REV{We study the problem of how well a tree metric is able to preserve the sum of pairwise distances of an arbitrary metric.
}
This problem is closely related to low-stretch metric embeddings and is interesting by its own flavor from the line of research proposed in the literature.

As the structure of a tree imposes great constraints on the pairwise distances,
any embedding of a metric into a tree metric is known to have maximum pairwise stretch of $\Omega(\log n)$. 
\REV{We show, however, from the perspective of average performance, there exist tree metrics which preserve the sum of pairwise distances of the given metric up to a small constant factor, for which we also show to be no worse than twice what we can possibly expect.
}
The approach we use to tackle this problem is more direct compared to a previous result of~\cite{Abraham:2006:AME:1132516.1132557}, and also leads to a provably better guarantee.
Second, when the given metric is extracted from a Euclidean point set of finite dimension $d$, we show that there exist spanning trees of the given point set such that the sum of pairwise distances is preserved up to a constant which depends only on $d$.
Both of our proofs are constructive.
The main ingredient in our result
is a special point-set decomposition which relates two seemingly-unrelated quantities.
\end{abstract}


\section{Introduction}

The problem of approximating a given metric by a metric which is structurally simpler has been a central issue to the theory of finite metric embedding and has been studied extensively in the past decades.
A particularly simple metric of interest, which also favors from the algorithmic perspective, is a tree metric. By a tree metric we mean a metric induced by the shortest distances
between pairs of points in
a tree containing the given points.
Generally we would require the distances in the given metric not to be underestimated
in the target metric,
which is crucial for most of the applications, and we would like to bound the increase of the distances, distortion, or stretch, from above.
See~\cite{Abraham:2007:EMU:1283383.1283437,Bartal:1996:PAM:874062.875536,Bartal:2003:MRP:780542.780610,Fakcharoenphol:2003:TBA:780542.780608}.
On the other hand, a similar and equally important problem in network design is to find a tree spanning the network, represented by a graph, that provides a good approximation 
to the shortest path metric defined in the graph~\cite{Abraham:2007:EMU:1283383.1283437,Alon:1995:GGA:205998.206008,Elkin:2005:LST:1060590.1060665}.

%
%

Let $\MC{M}=(V,d)$ and $\MC{M}^\prime = (V,d^\prime)$ be two metrics over the same point set $V$ such that $d^\prime(u,v) \ge d(u,v)$ for all $u,v \in V$.
For each $u,v \in V$, let $stretch(u,v)=d^\prime(u,v)/d(u,v)$ be the pairwise stretch, or distortion, between the pair $u$ and $v$.
Different notions have been suggested to quantify how well the distances of $\MC{M}$ are preserved in $\MC{M}^\prime$, e.g.,
\begin{enumerate}
\item
Maximum pairwise stretch~\cite{Narasimhan:2007:GSN:1208237}, defined by $\max_{u,v \in V}stretch(u,v)$, which is closely related to the extensively studied \emph{Spanner} problems.

\item 
Average pairwise stretch~
\cite{Abraham:2007:EMU:1283383.1283437,Elkin:2005:LST:1060590.1060665}, defined by $\BIGP{\sum_{u,v\in V}stretch(u,v)} / \binom{|V|}{2}.$

\item
Distance-weighted average stretch~\cite{Johnson1978,Rabinovich03onaverage,Wu:1999:PAS:337729.337748},
defined as $$\frac{1}{\sum_{u,v\in V}d(u,v)}\sum_{u,v \in V}d(u,v)\cdot stretch(u,v) = \frac{\sum_{u,v \in V}d^\prime(u,v)}{\sum_{u,v\in V}d(u,v)}.$$
This measure makes sense in real-time scenarios when it is less desirable and more costly to raise the distances of distant pairs than that of close pairs.
For example, the effect of raising the delay of a pair from 2 seconds to 10 seconds is less tolerable than raising the delay of another pair from 20 ms to 100 ms.
Throughout this paper we will also refer to the sum of pairwise distances as the routing cost following the terminology used in the literature.
\end{enumerate}
%

\smallskip

In this work, we address the problem of how well a tree is able to preserve the sum of pairwise distances, or, the distance-weighted average stretch, of an underlying metric.
To be more precise, let $\MC{M} = (V,d)$ and $\MC{M}^\prime = (V^\prime,d^\prime)$ be two metrics. We say that $\MC{M}^\prime$ dominates $\MC{M}$ if $V^\prime \supseteq V$ and for all $u,v \in V$, we have $d^\prime(u,v) \ge d(u,v)$.
We consider the following two problems.


\begin{problem} \label{prob_metric}
Let $\MC{M} = (V,d)$ be a given metric and $\MC{D}(\MC{M})$ be the set of dominating tree metrics of $\MC{M}$. What is $$\inf_{(V^\prime,d^\prime) \in \MC{D}(\MC{M})}{\frac{\sum_{u,v \in V}d^\prime(u,v)}{\sum_{u,v \in V}d(u,v)}} \enskip ?$$
\end{problem}

\begin{problem} \label{prob_graph}
Let $V$ be a set of points in $\MC{R}^d$, 
\REV{
$\BIGC{\overline{u,v}}$ be the straight-line distance between two points $u,v\in V$,
}
$\MC{ST}(V)$ be the set of spanning trees of $V$, and $d_\MC{T}$ be the distance function of $\MC{T}$, for any $\MC{T} \in \MC{ST}(V)$. What is
$$\inf_{\MC{T} \in \MC{ST}(V)}{\frac{\sum_{u,v \in V}d_\MC{T}(u,v)}{\sum_{u,v \in V}\BIGC{\overline{u,v}}}} \enskip ?$$
\end{problem}


We remark on Problem~\ref{prob_graph} that, although we can consider the Euclidean metric extracted from $V$ as we did in Problem~\ref{prob_metric}, dominating tree metrics of it do not necessarily correspond to any spanning tree of $V$. In fact, if we apply the approaches for Problem~\ref{prob_metric} directly, the lack of balance guarantee in each partition can make the resulting pairwise distances arbitrary large.

\smallskip


Embedding metrics into tree metrics was introduced in the context of probabilistic embedding by Alon et al.,~\cite{Alon:1995:GGA:205998.206008}.
What follows was a series of notable work.
%
%
Bartal~\cite{Bartal:1996:PAM:874062.875536} considered probabilistic embeddings and proved that any metric can be probabilistically approximated by tree metrics with expected maximum distortion $O(\log^2n)$. This result was later improved to $O(\log n\log\log n)$~\cite{Bartal:1998:AAM:276698.276725}.
Bartal also observed that any probabilistic embedding into a tree has distortion at least $\Omega(\log n)$. This gap was closed by Fakcharoenphol et al.,~\cite{Fakcharoenphol:2003:TBA:780542.780608}, who showed that for any metric, there exists tree metrics with $O(\log n)$ distortion.

\begin{problem} \label{prob_dual}
Given a metric $M=(V,d)$ and a weight function $w: V \times V \rightarrow \MC{R}^+$, find a dominating tree metric $T$ of $M$ such that $\sum_{u,v\in V}w_{uv}\cdot d_T(u,v) \le \alpha \sum_{u,v\in V}w_{uv}\cdot d(u,v).$
\end{problem}

As Charikar et al.,~\cite{Charikar:1998:AFM:795664.796406} showed by linear program
duality that computing probabilistic embeddings of a given metric and
Problem~\ref{prob_dual} described above are in fact dual problems,
the series of work led by Bartal~\cite{Bartal:1996:PAM:874062.875536,Bartal:1998:AAM:276698.276725,Elkin:2005:LST:1060590.1060665,Fakcharoenphol:2003:TBA:780542.780608} has provided improved approximation results for a large set of problems, including \emph{buy-at-bulk network design}, \emph{vehicle routing}, \emph{metric labeling}, \emph{group Steiner tree}, \emph{Minimum cost communication network}. Refer to~\cite{Bartal:1998:AAM:276698.276725,Charikar:1998:AFM:795664.796406} for more detail and applications.


Kleinberg, Slivkins, and Wexler~\cite{Kleinberg:2009:TEU:1568318.1568322} initiated the study of partial embedding and scaling distortion, which can be regarded as embedding with relaxed guarantees. 
In a series of following work, Abraham et al.,~\cite{Abraham:2005:MER:1097112.1097448,Abraham:2006:AME:1132516.1132557} proved that any finite metric embeds probabilistically in a tree metric such that the distortion of $(1-\epsilon)$ portion of the pairs is bounded by $O(\log\frac{1}{\epsilon})$, for any $0<\epsilon<1$. They also observed a lower bound of $\Omega(\sqrt{1/\epsilon})$, which is closed by Abraham et al., in~\cite{Abraham:2007:EMU:1283383.1283437}.
In particular, Abraham et al.,~\cite{Abraham:2006:AME:1132516.1132557} 
showed that any metric can be probabilistically embedded into a tree metric such that the ratio between the expected sum of pairwise distances is $O(\log\Phi)$, where $\Phi$ is the effective aspect ratio of given distribution.
This provides an upper-bound to Problem~\ref{prob_metric} we considered. However, the guarantee they provided is loose due to the constant inherited from the guarantee on scaling distortion. 
\REV{See also~\cite{Abraham:2005:MER:1097112.1097448,abrahambartal05,Abraham:2007:EMU:1283383.1283437}.
}
%
%
Rabinovich~\cite{Rabinovich03onaverage} showed that it is possible to embed certain special graph metrics into real line such that distance-weighted average stretch is bounded by a constant.

On the other hand, for approximating arbitrary graph metrics by their spanning trees, a simple $\Omega(n)$ lower bound in terms of maximum stretch is known for $n$-cycles
\cite{Rabinovich96lowerbounds}.
Alon, Karp, Peleg, and West~\cite{Alon:1995:GGA:205998.206008} considered a distribution over spanning trees and proved an upper bound of $2^{O\BIGP{\sqrt{\log n\log\log n}}}$ on the expected distortion.
Elkin et al.,~\cite{Elkin:2005:LST:1060590.1060665} showed how a spanning tree with $O(\log^2n\log\log n)$ average stretch (over the set of edges) can be computed in polynomial time. 
In terms of average pairwise stretch, Abraham et al.,~\cite{Abraham:2007:EMU:1283383.1283437} showed the existence of a spanning tree such that, for any $0 < \epsilon < 1$, the distortion of an $(1-\epsilon)$ fraction of the pairs is bounded by $O(\sqrt{1/\epsilon})$. 
Note that this implies an $O(1)$ average pairwise stretch.
Smid~\cite{smid09} gave a simpler proof for this result when the metric is Euclidean.

In terms of sum of pairwise distances in graphs (routing cost), 
Johnson et al.,~\cite{Johnson1978} showed that computing the spanning tree of minimum routing cost is NP-hard. Polynomial time approximations as well as approximation schemes have been proposed by Wong~\cite{Wong-1980} and Wu et al.,~\cite{Wu:1999:PAS:337729.337748}.
Despite the efforts devoted, however, no general guarantees have been made on the ratio between the routing cost of the optimal spanning tree and that of the underlying graphs.
Other reasonable variations have been considered as well, i.e., \emph{sum-requirement routing trees}, \emph{product-requirement routing trees}, 
and \emph{multi-sources routing trees}~
\cite{Wu:2004:AAO:1039240.1039244,Wu:2000:PTA:346489.346493,Yewu_lightgraphs}.

\paragraph*{Our Contribution}

In this work, we take a different approach to tackle Problem~\ref{prob_metric} directly and obtain a provably small upper-bound. Specifically, we adopt the notion of \textit{hierarchically well-separated trees} (HSTs), introduced by Bartal~\cite{Bartal:1998:AAM:276698.276725} and Fakcharoenphol~\cite{Fakcharoenphol:2003:TBA:780542.780608}, 
and show that, for any given metric $\MC{M}$, there exists
a 2-HST, $\MC{M}^\prime$, such that the distance-weighted average
stretch of $\MC{M}^\prime$ is bounded by $14.24$.
The main ingredient of this result is a special point-set decomposition 
which relates two seemingly-unrelated quantities, namely, the diameter of the point set and the sum of pairwise distances between two separated subsets.
%

If we do not require HSTs, it is also possible to apply our technique and construct the so-called \textit{ultra-metrics}, which is introduced by Abraham~\cite{Abraham:2007:EMU:1283383.1283437} and Bartal~\cite{bartal:graph04}, with a similar stretch, $3.56$.
\REV{
This provides a better and explicit guarantee than that provided in~\cite{Abraham:2006:AME:1132516.1132557} (from $\ge 64$).
For the negative side, we show that there exist metrics for which no dominating tree metrics can preserve the sum of pairwise distances to a factor better than $2$.
This shows that our result is within twice the best one can achieve.
}

As a side-product, we prove the existence of spanning trees with $O(d\sqrt{d})$ distance-weighted average stretch for any point set in Euclidean space $\MC{R}^d$.
%
To this end, we use our point-set cutting lemma to decompose the points recursively. In order to guarantee a constant blow-up in the diameter of the spanning tree, however, instead of allowing arbitrary cuts, we show that it is always possible to make a balanced decomposition such that the diameters of the partitioned sets stay balanced.
\REV{
Our result provides a good guarantee when the dimension of the given Euclidean graph is low, which is true for most communication network.
Although it is possible to apply the framework of~\cite{abrahambartal05,Abraham:2007:EMU:1283383.1283437} to obtain a spanning tree of constant distance-weighted average stretch, the constant hidden inside is huge ($>10^5$) that makes it practically less useful.
}
Both of our proofs are constructive.

\section{Preliminary} \label{preliminary}

First we define some notation that will be used throughout this paper.
Let $(M,d)$ be a finite metric space,
where $M$ is the set of vertices and $d$ is the distance function.
Without loss of generality, we shall assume that the smallest distance
defined by $d$ is strictly more than $1$. 
Let $X \subseteq M$ be a subset of $M$. 
The radius of $X$ with respect to a specific element $y\in X$ is defined to be $\Delta_y(X) = \max_{z\in X}d(y,z)$.
The \emph{diameter} of $X$ is defined to be $\Delta(X) = \max_{y\in X}\Delta_y(X)$.
%
%
For any $r \ge 0$, an $r$-net decomposition of $(M,d)$ is a partition of $M$ into clusters, where each cluster, say $\MC{C}$, has radius at most $r$ with respect to a certain vertex $u \in \MC{C}$.


\begin{definition}[Hierarchical net decomposition]
Let $(M,d)$ be a metric and $\delta = \left\lceil\log_2\Delta(M)\right\rceil$.
A hierarchical net decomposition of $(M,d)$ is a sequence of $\delta+1$ nested net decompositions $D_0, D_1, \ldots, D_\delta$ such that
\begin{itemize}
	\item
		$D_\delta = \left\{M\right\}$ is the trivial partition that puts all vertices in a single cluster.
	\item
		$D_i$ is a $2^i$-net decomposition and a refinement of $D_{i+1}$.
\end{itemize}
\end{definition}

A laminar family $\MC{F} \subseteq 2^M$ of a set $M$ is a family of subsets of $M$ such
that for any $A,B \in \MC{F}$, we have either $A \subseteq B$, $B \subseteq A$, or $A \cap B = \phi$.
Clearly a hierarchical net decomposition defines a laminar family and naturally
corresponds to a rooted tree, for which is referred to as a hierarchically well-separated tree (HST), as follows. Each set $S$ in the laminar family is a node in the tree, and the children of the node corresponding to $S$ are the nodes corresponding to maximal subsets of $S$ in the family.

The distance function on this tree is defined as follows. The links from
a node $S$ in $D^i$ to each of its children in the tree have length equal to
$2^{i-1}$.
This induces a distance
function $d_T$ on $M$, where $d_T(u,v)$ is equal to the length of the shortest
path distance in $T$ from node $u$ to node $v$.

\begin{definition}[Ultrametric]
An ultrametric $M$ is a metric space $(M, d)$ whose elements are the leaves of a rooted labelled tree $T$ such that the following is met.
Each node $v \in T$ is associated with a label $\ell(v) \ge 0$ such that if $u \in T$ is a descendant of $v$ then
$\ell(u)\le \ell(v)$ and $\ell(v)=0$ if and only if $v$ is a leaf node. The distance between leaves $u, v \in M$
is defined as $d(u,v)=\ell(lca(u,v))$, where $lca(u,v)$ is the least common ancestor of $u$ and $v$ in $T$.
\end{definition}

Note that, under this definition, the metric extracted from a hierarchically well-separated tree is also an ultrametric.

\begin{definition}[Centripetal metric]
Given a metric $(M,d)$ and a vertex $x \in M$, we define the centripetal metric
$(M,d_x)$ of $(M,d)$ with respect to $x$ as $d_x(u,v) = \BIGLR{|}{|}{d(u,x) - d(v,x)}$.
\end{definition}

For any metric $(X,d)$, we denote by $\MC{R}_d(X) = \sum_{u,v \in X}d(u,v)$ the sum of pairwise distances over $X$. 
Let $P,Q \subset X$ be subsets of $X$ such that $P\cap Q = \phi$, we define $\MC{R}_d(P,Q) = \sum_{u\in P, v\in Q}d(u,v)$ to be the
sum of pairwise distances between $P$ and $Q$.
The subscripts $d$ will be omitted when there is no confusion.
Clearly, $\MC{R}(X)$ decomposes into $\MC{R}(P) + \MC{R}(Q) + \MC{R}(P,Q)$ when $P$ and $Q$ form a partition of $X$.

Consider the Euclidean space of finite dimension $d$. 
A hyper-rectangle is defined to be the Cartesian product of $d$ closed intervals,
which we will denote by $[a_1, b_1] \times [a_2, b_2] \times \ldots\times [a_d,b_d]$.
Given a hyper-rectangle $R = [a_1, b_1] \times [a_2, b_2] \times \ldots\times [a_d,b_d]$,
we denote by $\MC{L}_i(R)$ the side length of $R$ along the $i^{th}$ dimension, which is $b_i-a_i$,
and $\MC{L}_{max}(R) = \max_{1\le i\le d}\MC{L}_i(R)$.
For a point set $S \in \MC{R}^d$, we define its bounding box, denoted by $\MC{B}(S)$, to be the smallest hyper-rectangle that contains $S$.


\section{Approximating Arbitrary Metrics}

Given a metric $M=(V,d)$, we describe in this section how a 
tree metric with small constant distance-weighted average stretch
can be computed.
%

\subsection{The Algorithm} \label{subsection_arbitrary_metric_algo}

We describe an algorithm to decompose $M$ and define a hierarchical net decomposition. 
%
The algorithm runs in $\delta = \BIGLR{\lceil}{\rceil}{\log_2\Delta(V)}$ iterations.
Initially, we have $i = \delta$ and 
the trivial partition $D_\delta = \BIGLR{\{}{\}}{M}$.
In each of the following iteration, we decrease the value of $i$ by one and 
compute $D_i$ from $D_{i+1}$ as follows.

\begin{figure}[t]
\centering
\includegraphics[scale=0.9]{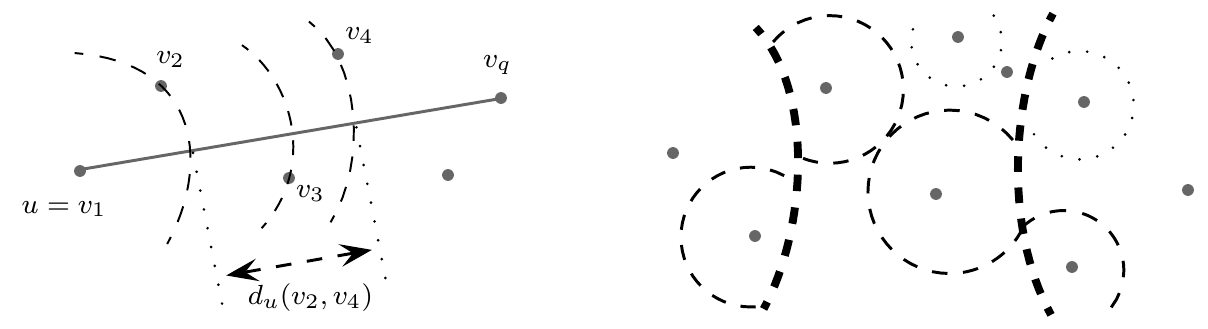}
\caption{(a) An illustration of the centripetal metric with respect to a vertex $u$. (b) A hierarchical decomposition of the points. 
}
\label{figure_centripetal}
\end{figure}

For each non-singleton cluster in $D_{i+1}$, say $\MC{P}$, 
we compute a $2^i$-cut decomposition $\MC{C}(\MC{P})$ of $\MC{P}$
%
by repeatedly decomposing $\MC{P}$ by the process described below until the diameter of each clusters in the refinement falls under $2^i$.

Let $\MC{Q}$ be a cluster in the refinement of $\MC{P}$ such that $\Delta(\MC{Q}) \ge 2^i$.
We pick a vertex $u \in \MC{Q}$ such that $\Delta_u(\MC{Q}) = \Delta(\MC{Q})$.
Then we consider the centripetal metric of $\MC{Q}$ with respect to $u$.
Let $v_1, v_2, \ldots, v_q$ be the set of vertices of $\MC{Q}$ such that
$d(u,v_1) \le d(u,v_2) \le \ldots \le d(u,v_q)$.
For $1\le i \le q-1$, we denote
$\sum_{1\le j\le i}\sum_{i<k\le q}d_u(v_j,v_k)$ by $\MC{RC(i)}$.
Literally, $\MC{RC}(i)$ corresponds to the sum of pairwise distances, or, the interaction, between 
$\{v_1, v_2, \ldots, v_i\}$ and $\{v_{i+1}, v_{i+2}, \ldots, v_q\}$.
Let $p$, $1\le p < q$, be the index such that 
$\frac{p\cdot (q-p)\cdot \Delta(\MC{Q})}{\MC{RC}(p)}$
is minimized.
%
%
%
We create a new cluster in the refinement of $\MC{P}$ containing the vertices $\{v_1, v_2, \ldots, v_p\}$ and
let $\MC{Q} \leftarrow \MC{Q} \backslash \{v_1, v_2, \ldots, v_p\}$.
This process is repeated until all the clusters in the refinement of $\MC{P}$ have diameter less than $2^i$.
$D_i$ is defined to be the union of the refinements of non-singleton clusters of $D_{i+1}$.
%
%
A high-level description of this algorithm can be found in the appendix.

\subsection{Analysis} \label{subsection_arbitrary_metric_analysis}

First we argue that the algorithm computes a dominating tree metric.
%
%
Let $T$ be the tree corresponding to the hierarchical net decomposition constructed by our algorithm and $d_T$ be the distance function induced by $T$. 
For any non-singleton cluster $\MC{P}$ in $D_i$ and $u,v\in \MC{P}$, we have $d(u,v) \le \Delta(\MC{P}) < 2^i$ by the definition of hierarchical net decomposition, and $d_T(u,v) \le 2\cdot \sum_{0\le j < i} 2^j < 2^{i+1}$ by the construction of the tree metric.
Therefore, $(T,d_T)$ is a dominating tree metric of $M$.

In the following, we will show that
$\MC{R}(T) \le 4\cdot\frac{210}{59}\cdot \MC{R}(M)$.
To this end, we prove that, for any partition of a cluster $\MC{Q}$ into, say $\MC{Q}_1$ and $\MC{Q}_2$ such that $u \in \MC{Q}_1$, we performed in our algorithm,
we have 
\begin{align}
\BIGC{\MC{Q}_1}\cdot \BIGC{\MC{Q}_2}\cdot \Delta(\MC{Q}) \le \frac{210}{59}\cdot \MC{R}(\MC{Q}_1, \MC{Q}_2). \label{ieq_rc_partition}
\end{align}
Let $T[\MC{Q}]$, $T[\MC{Q}_1]$, and $T[\MC{Q}_2]$ denote the subtree of $T$ corresponding to $\MC{Q}$, $\MC{Q}_1$, and $\MC{Q}_2$, respectively.
As a consequence to $(\ref{ieq_rc_partition})$, we have $\MC{R}(T_{\MC{Q}_1}, T_{\MC{Q}_2}) \le \BIGC{\MC{Q}_1}\cdot \BIGC{\MC{Q}_2}\cdot 2^{i+1} \le 4\cdot \BIGC{\MC{Q}_1}\cdot \BIGC{\MC{Q}_2}\cdot \Delta(\MC{Q}) \le 4\cdot\frac{210}{59}\cdot \MC{R}(\MC{Q}_1, \MC{Q}_2)$.
Since $\max\BIGLR{\{}{\}}{\BIGC{\MC{Q}_1}, \BIGC{\MC{Q}_2}} < \BIGC{\MC{Q}}$,
by an inductive argument we have $\MC{R}(T_\MC{Q}) = \MC{R}(T_{\MC{Q}_1}) + \MC{R}(T_{\MC{Q}_2}) + \MC{R}(T_{\MC{Q}_1}, T_{\MC{Q}_2}) \le 4\cdot\frac{210}{59}\cdot\BIGP{\MC{R}(\MC{Q}_1) + \MC{R}(\MC{Q}_2) + \MC{R}(\MC{Q}_1, \MC{Q}_2)} = 4\cdot\frac{210}{59}\cdot\MC{R}(\MC{Q})$. This holds for all cluster $\MC{Q}$, including the trivial cluster in $D_\delta$. Therefore $\MC{R}(T) \le 4\cdot\frac{210}{59}\cdot \MC{R}(M)$.

\medskip

It remains to prove the inequality $(\ref{ieq_rc_partition})$. Let $\{v_1, v_2, \ldots, v_q\}$ be the set of vertices of $\MC{Q}$ such that $d(u,v_1)\le d(u,v_2)\le \ldots \le d(u,v_q)$.
%
Consider the following random distribution defined over $\beta \in \BIGLR{\{}{\}}{\BIGLR{\lceil}{\rceil}{\frac{q}{4}}, \BIGLR{\lceil}{\rceil}{\frac{q}{4}}+1,\ldots,\BIGLR{\lfloor}{\rfloor}{\frac{3q}{4}}}$.

\vspace{-15pt}
\begin{align*}
Pr\BIGLR{[}{]}{\beta = i} = \frac{\MC{RC(i)}}{\sum_{\frac{q}{4}\le i \le \frac{3q}{4}}\MC{RC(i)}} 
\end{align*}

Let us first derive a lower bound on $\sum_{\frac{q}{4}\le i \le\frac{3q}{4}}\MC{RC}(i)$, which is the total amount of interaction when cutting at the central $\frac{q}{2}$ intervals.
Due to space limit, preliminary material as well as proofs to the following lemmas are moved to the appendix for further reference.

\begin{lemma} \label{lemma_central_interaction_lower_bound}
We have $$\sum_{\frac{q}{4} \le i \le \frac{3}{4}q}\MC{RC}(i)
\ge \BIGP{\frac{3}{32}q^3 + \frac{q}{2}\cdot\sum_{\frac{q}{6}q\le i \le \frac{q}{4}}i} \cdot \sum_{\frac{q}{3}\le k\le \frac{2q}{3}}\ell_k$$
\end{lemma}

The following lemma proves the existence of a good cut and 
$(\ref{ieq_rc_partition})$.

\begin{lemma} \label{lemma_expectation_upper_bound}
We have
{\small
$$\min\BIGLR{\{}{\}}{E\BIGLR{[}{]}{\frac{\beta\cdot(q-\beta)\cdot\Delta(\MC{Q})}{\MC{RC}(\beta)}}, \min_{1\le \gamma\le \frac{q}{3}}\BIGLR{\{}{\}}{\frac{\gamma\cdot(q-\gamma)\cdot\Delta(\MC{Q})}{\MC{RC}(\gamma)}, \frac{\gamma\cdot(q-\gamma)\cdot\Delta(\MC{Q})}{\MC{RC}(q-\gamma)}}} \le \frac{210}{59}.$$
}
\end{lemma}



%
%
As a side-product, we have the following lemma, which 
states the existence of good cuts for any given point set and the correctness of inequality~$(\ref{ieq_rc_partition})$.

\begin{lemma}[$1$-Dimensional Point Set Cutting Lemma] \label{lemma_1_d_cutting}
Given a set of real numbers $A = \left\{a_1, a_2, \ldots, a_n\right\}$, $a_1 \le a_2 \le \ldots \le a_n$, there exists a cutting point $z \in R$ with
$a_1 < z < a_n$ such that the following holds.
$$L_A(z) \cdot \left(n-L_A(z)\right) \cdot \Delta \le \delta_0 \cdot \sum_{1\le i\le L_A(z)}\enskip\sum_{L_A(z)<j\le n}(a_j-a_i),$$
where $L_A(z) = \BIGC{\{a\in A: a < z\}}$ is the number of elements in A that are smaller than $z$,
$\Delta = a_n-a_1$ is the diameter of $A$, and $\delta_0 \le \frac{210}{59}$ is a constant.
\end{lemma}


\subsection{Lower Bound}  \label{subsection_lower_bounds}

In the following, we derive a lower bound to Problem~\ref{prob_metric} we considered throughout this section. This is done by linking the basic structure of any optimal dominating tree metric to our point set cutting lemma, followed by deriving an upper bound to the performance of any cut.

Let $\MC{A}=\{a_1, a_2, \ldots, a_n\}$ be a set of numbers, where $a_i = i$ for all $1\le i\le n$,
and $(\MC{A},d)$ be the corresponding metric extracted from $\MC{A}$.
Let $(T,d_T)$ be an optimal ultra-metric embedding of $\MC{A}$ in terms of distance-weighted average stretch.
Without loss of generality, we can assume that $T$ is a binary tree.
Otherwise, we can always create dummy nodes to make $T$ binary without changing its sum of pairwise distances.
The following lemma characterizes the structure of $T$.

\begin{lemma} \label{lemma_ultrametric_lower_bound_partition}
Let $T_L$ and $T_R$ be the left-subtree and the right-subtree of $T$ such that $a_1 \in T_L$.
Then, there exists an integer $k$, $1\le k<n$, such that $T_L$ is an ultra-metric containing $\{a_1, a_2, \ldots, a_k\}$ and $T_R$ is an ultra-metric containing $\MC{A} \backslash \{a_1, a_2, \ldots, a_k\}$.
\end{lemma}

Therefore, to obtain a lower bound on the distance-weighted average stretch of any dominating tree metric of $\MC{A}$, it suffices to consider the quality of the best cut we can possibly achieve on $\MC{A}$.

\begin{lemma} \label{lemma_cutting_lemma_lower_bound}
Let $\delta_0$ be a constant such that our point set cutting lemma holds, then $\delta_0 \ge 2$.
\end{lemma}

By Lemma~\ref{lemma_ultrametric_lower_bound_partition} and Lemma~\ref{lemma_cutting_lemma_lower_bound}, we obtain the following bound as claimed.
%

\begin{corollary} \label{cor_dominating_metric_lower_bound}
Let $\MC{M} = (V,d)$ be a given metric and $\MC{D}(\MC{M})$ be the set of dominating tree metrics of $\MC{M}$. Then $$\inf_{(V^\prime,d^\prime) \in \MC{D}(\MC{M})}{\frac{\sum_{u,v \in V}d^\prime(u,v)}{\sum_{u,v \in V}d(u,v)}} \ge 2.$$
\end{corollary}

\section{Approximating Euclidean Metrics by Their Spanning Trees} \label{section_euclidean_graph}

In this section, we show how a spanning tree of small constant distance-weighted average stretch
for a Euclidean graph can be computed in polynomial time.
The basic idea is to 
extend the previous point-set decomposition.
In order to guarantee
a constant blow-up in the diameter of the resulting spanning tree, we cannot allow
the cut to be made at arbitrary positions. 
Instead, we restrict each cut to be made within the central $(1-2\alpha)$ portion
along the longest side of its bounding box,
where $\alpha$ is a constant chosen to be $\frac{1}{4}$.
This guarantees a balanced partition, an exponentially decreasing size of the bounding boxes, and a constant
blow-up of the diameter of the resulting spanning tree.
This is crucial in the analysis, as we need a tight diameter in order to provide a
good upper-bound on the interaction between pairs separated by our cuts.
On the other hand, we also have to guarantee the existence of good cuts in the central $(1-2\alpha)$ portion so that the overall interaction stays bounded.

\smallskip

Given a set of points $\MC{P}$ in the Euclidean space $\MC{R}^d$ of finite dimension, our algorithm recursively computes a rooted tree $\MC{T}$ with root $r$ as follows. Let $\MC{B}(\MC{P})$ be the bounding box of $\MC{P}$, and $k$ be the index of dimension such that $\MC{L}_k(\MC{B}(\MC{P})) = \MC{L}_{max}(\MC{B}(\MC{P}))$.
We consider the projection of the points to the $k^{th}$-axis, and let $a_1, a_2, \ldots, a_n$, $a_1 \le a_2 \le \ldots \le a_n$, be the corresponding coordinates.
We apply our linear time algorithm\footnote{This algorithm is moved to \S~\ref{section_linear_time_cut} for further reference due to space limit.} to compute a decomposition for which the cut is restricted to be made inside the central $(1-2\alpha)$ portion, 
$\BIGLR{[}{]}{\alpha\cdot(a_1+a_n), (1-\alpha)\cdot(a_1+a_n)}$.
See also Fig.~\ref{figure_bounding_box_cutting}~(a).
Let $\MC{P}_1$ and $\MC{P}_2$ be the corresponding partitioned subsets of points.
We compute recursively the two rooted trees for $\MC{P}_1$ and $\MC{P}_2$, denoted by $\MC{T}_1$ with root $r_1$ and $\MC{T}_2$ with root $r_2$.
The tree $\MC{T}$ is constructed by joining $r_1$ and $r_2$, and the root of $\MC{T}$ is chosen to be $r_1$.
A high-level description of our algorithm is provided in the appendix. 

\begin{figure}[h]
\hspace{-45pt}
\includegraphics[scale=.8]{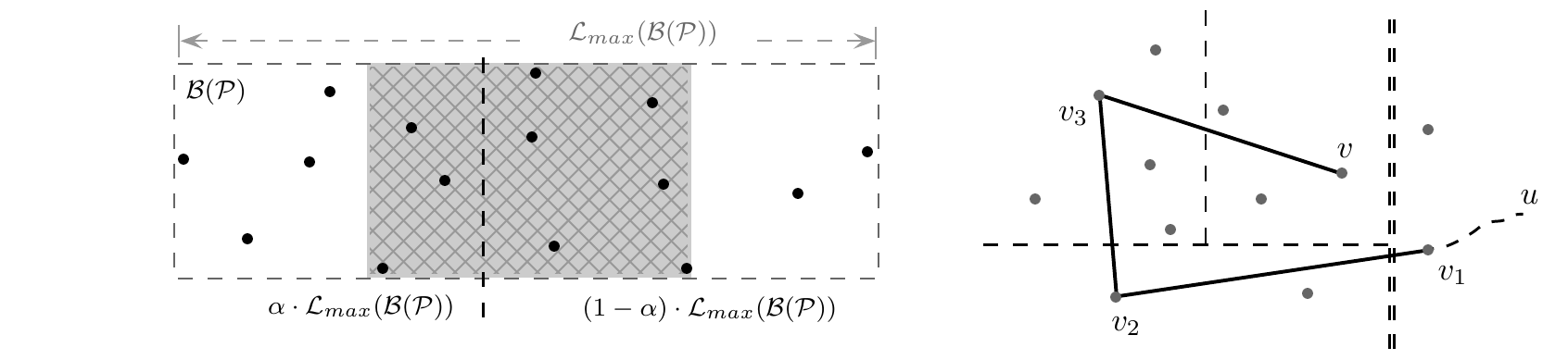}
\vspace{-20pt}
\caption{(a) The vertical cut is restricted to be placed in the central $(1-2\alpha)$ portion along the longest side of the bounding box.
(b) A possible decomposition and the $u-v$ path in the resulting tree. 
}
\label{figure_bounding_box_cutting}
\end{figure}

In the following lemma, we show that, in exchange of certain penalty in the performance factor that is  inverse proportional to the length of the interval to which the cut is restricted, we can always guarantee a good and balanced decomposition.

\begin{lemma}[Constrained Point Set Cutting Lemma] \label{lemma_1_d_weighted_cutting}
Given a set of real numbers $A = \left\{a_1, a_2, \ldots, a_n\right\}$, $a_1\le a_2\le \ldots \le a_n$
and an interval $\MC{I} = [\ell,r]$ such that $\MC{I} \subseteq [a_1,a_n]$, there exists a cutting point $z \in \MC{I}$ such that the following holds.
$$L_A(z) \cdot \BIGP{n-L_A(z)} \cdot \BIGC{\MC{I}} \le \delta_0\cdot \sum_{1\le i\le L_A(z)}\enskip \sum_{L_A(z)<j\le n}(a_j-a_i),$$
where $L_A(z) = \BIGC{\{a\in A: a < z\}}$ is the number of elements in A that are smaller than $z$
and $\delta_0 \le \frac{210}{59}$ is a constant.
\end{lemma}

In the following, we state the theorem and leave the rest detail in the appendix for further reference.

\begin{theorem} \label{thm_euclidean_tree}
Given a set of points $\MC{P}$ in $\MC{R}^d$, 
we can compute in polynomial time a spanning tree $\MC{T}$ of $\MC{P}$ such that the distance-weighted
average stretch of $\MC{T}$ with respect to $\MC{P}$ is at most $16\delta_0\cdot d\sqrt{d}$, where $\delta_0 \le \frac{210}{59}$ is the constant in our point set cutting lemma.
\end{theorem}



\section{Discussion and Open Problems}  \label{section_conclusion}

We conclude with some remarks and conjectures. 
\REV{
In this work, we provided both an upper bound and a lower bound to Problem~\ref{prob_metric}.
We conjecture the lower bound of two we provided to be tight.
}
%
On the other hand, we also conjecture that similar result holds for approximating arbitrary graph metrics by their spanning trees.
However, as it seems not promising to guarantee the quality of the best cut for arbitrarily small restricted intervals, none of known graph decomposition techniques helps and either more powerful decomposition schemes or new techniques are expected.


\bibliographystyle{plain}

\small
\bibliography{routing_cost_spanning_tree}

\newpage

\begin{appendix}

\section{Approximating Arbitrary Metrics} \label{apx_sec_metric}


\subsection{The Algorithm}

\begin{figure*}[h]
\vspace{-15pt}
\rule{\linewidth}{0.2mm}
\smallskip
{{\sc Algorithm} {\em Hierarchical-Net-Decomposition$(V,d)$}}

\begin{algorithmic}[1]
\STATE $D_\delta \leftarrow \{V\}$, $i \leftarrow \delta -1$.
\WHILE{$i\ge 0$ and $D_{i+1}$ has non-singleton clusters}
	\FORALL{non-singleton cluster $\MC{P}$ in $D_{i+1}$}
		\STATE $\MC{C}(\MC{P}) \leftarrow \{\phi\}$, \enskip $\MC{S} \leftarrow \BIGLR{\{}{\}}{\MC{P}}$.
		\WHILE{$\MC{S} \neq \phi$}
			\STATE Let $\MC{Q}$ be an arbitrary cluster in $\MC{S}$.
			\IF{$\Delta(\MC{Q}) < 2^i$}
				\STATE Add $\MC{Q}$ to $\MC{C}(\MC{P})$ and remove $\MC{Q}$ from $\MC{S}$. 
			\ELSE
				\STATE Let $u\in \MC{Q}$ be a vertex such that $\Delta_u(\MC{Q}) = \Delta(\MC{Q})$.
				\STATE Let $v_1, v_2, \ldots, v_q$ be the set of vertices in $\MC{Q}$ such that $d(u,v_1) \le d(u,v_2) \le \ldots \le d(u,v_q)$.
				\STATE Let $p$, $1\le p<q$, be the index such that $\frac{p\cdot (q-p)\cdot \Delta(\MC{Q})}{\MC{RC}(p)}$ is minimized.
				\STATE Let $\MC{Q}^\prime \leftarrow \BIGLR{\{}{\}}{v_1, v_2, \ldots, v_p}$, $\MC{S} \leftarrow \MC{S} \cup \{\MC{Q}^\prime\}$, and $\MC{Q} \leftarrow \MC{Q} \backslash \MC{Q}^\prime$.
			\ENDIF
		\ENDWHILE
		\STATE Let $\MC{C}(\MC{P})$ be the refinement clusters of $\MC{P}$ in $D_i$.
	\ENDFOR
	\STATE $i \leftarrow i-1$.
\ENDWHILE
\STATE Return the tree metric corresponding to $D_0, D_1, \ldots, D_\delta$.
\end{algorithmic}
\rule{\linewidth}{0.2mm}
\vspace{-10pt}
\caption{A high-level description of the algorithm.} 
\label{algorithm_random_partition}
\end{figure*}
\vspace{-20pt}

\subsection{Analysis}

\begin{ap_lemma}{\ref{lemma_central_interaction_lower_bound}}
We have 
$$\sum_{\frac{q}{4} \le i \le \frac{3}{4}q}\MC{RC}(i)
\ge \BIGP{\frac{3}{32}q^3 + \frac{q}{2}\cdot\sum_{\frac{q}{6}q\le i \le \frac{q}{4}}i} \cdot \sum_{\frac{q}{3}\le k\le \frac{2q}{3}}\ell_k$$
\end{ap_lemma}

Before proving Lemma~\ref{lemma_central_interaction_lower_bound}, 
let us derive a lower bound on the overall interaction
$\sum_{1\le i<q}\MC{RC}(i)$.
Recall that, 
$\MC{RC}(i) = \sum_{1\le j<i}\sum_{i<j\le q}d_u(v_j,v_k)$ and 
$d_u(v_j,v_k) = \BIGLR{|}{|}{d(u,v_j)-d(u,v_k)}$.
For convenience, we will denote by $\ell_k$ the quantity $d_u(v_k,v_{k+1})$,
which is exactly $d(u,v_{k+1}) - d(u,v_k)$, for each $1\le k < q$.

First, observe that, for each $j,k$ with $1<j < k<q$, we have
exactly $(k-j)$ duplications of the item
$d_u(v_j,v_k)$ in the summation
$\sum_{1 \le i < q}\MC{RC}(i)$, i.e., it appears exactly once 
in $\MC{RC}(i)$ for each $j\le i<k$.
Therefore, after re-arranging the items we have
$$\sum_{1 \le i<q}\MC{RC}(i) = \sum_{1\le k<q} k\cdot\sum_{1\le i \le q-k}
d_u(v_i,v_{i+k}).$$

Let $f(q) = \frac{q}{2}\sum_{1\le i\le \frac{q}{2}}d_u(v_i,v_{i+\frac{q}{2}})$ if $q$ is even and $f(q)=0$ otherwise.
Then
\begin{align*}
& \sum_{1\le k<q} k\cdot\sum_{1\le i \le q-k} d_u(v_i,v_{i+k}) \\
= & \sum_{1\le k < \frac{q}{2}}k \cdot\sum_{1\le i\le q-k} d_u(v_i,v_{i+k})
+\sum_{\frac{q}{2} < k < q}k \cdot\sum_{1\le i\le q-k} d_u(v_i,v_{i+k}) + f(q) \\
= & \sum_{1\le k < \frac{q}{2}}k \cdot\sum_{1\le i\le q-k} d_u(v_i,v_{i+k}) +
\sum_{1\le k < \frac{q}{2}} (q-k) \sum_{1\le i\le k} d_u(v_i,v_{i+q-k}) + f(q),
\end{align*}

where in the last inequality we substitute the variable $k$ by $q-k$.
%
By re-organizing and aligning the items from the above summation (see also Fig.~\ref{figure_alignment}), we have the following lemma.

\begin{figure}[t]
\centering
\includegraphics[scale=0.9]{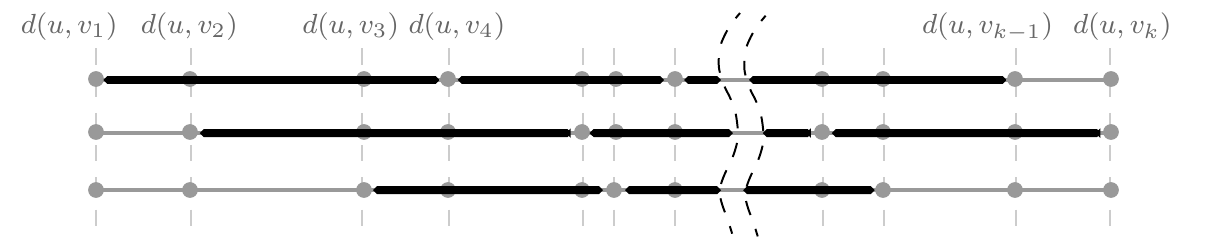}
\caption{Alignment of the intervals when $k=3$. The first group starts with $d(u,v_1)$ while the second and the third start with $d(u,v_2)$ and $d(u,v_3)$, respectively.}
\label{figure_alignment}
\end{figure}

\begin{lemma} \label{lemma_sum_fix_length_interval}
For $1\le k\le \BIGLR{\lfloor}{\rfloor}{\frac{q}{2}}$, we have
$$\sum_{1\le i\le q-k}d_u(v_i,v_{i+k})
= k\cdot\Delta(\MC{Q}) - \sum_{1\le i<k}(k-i)\cdot(\ell_i+\ell_{q-i})
= \sum_{1\le i\le k}d_u(v_i,v_{i+q-k})$$
\end{lemma}

\begin{proof}[Proof of Lemma~\ref{lemma_sum_fix_length_interval}]
We prove the first half of this lemma, $\sum_{1\le i\le q-k}d_v(v_{i+k},v_i)
= k\cdot\Delta(\MC{Q}) - \sum_{1\le i<k}(k-i)\cdot(\ell_i+\ell_{q-i})$. The second half, $\sum_{1\le i\le k}d_v(v_{i+q-k},v_i)=k\cdot\Delta(\MC{Q}) - \sum_{1\le i<k}(k-i)\cdot(\ell_i+\ell_{q-i})$, follows by a similar argument.
Consider the alignments of the set of intervals which spans exactly $k$ consecutive elements, that is, 
intervals $[d(u,v_i),d(u,v_{i+k})]$, for $1\le k\le \left\lfloor\frac{q}{2}\right\rfloor$.
We have exactly $k$ alignments, each starting with $\MC{I}_i$ for $1\le i\le k$.
See also Fig.~\ref{figure_alignment}.
This sums up to $k\cdot\Delta(\MC{Q})$, except for exactly $k-i$ times over-count of $\ell_i$ and $\ell_{q-i}$.
\hfill$\Box$
\end{proof}

We provide in the following lemma an overall estimate to the overall interaction, $\sum_{1\le i<q}\MC{RC}(i)$.

\begin{lemma} \label{lemma_counting_overall_routing_cost}
$$\sum_{1 \le i < q}\MC{RC}(i)
\ge \sum_{1\le k < \frac{q}{2}} q\cdot \sum_{\frac{q}{2} -k < i < \frac{q}{2}}i\cdot(\ell_k+\ell_{q-k}) + g(q),$$ where
$g(q) = q\cdot\sum_{1\le i<\frac{q}{2}}i\cdot\ell_\frac{q}{2}$ if $q$ is even and $g(q) = 0$ otherwise.
\end{lemma}

\begin{proof}[Proof of Lemma~\ref{lemma_counting_overall_routing_cost}]
By the above discussion and Lemma~\ref{lemma_sum_fix_length_interval}, we have
\begin{align*}
& \sum_{1 \le i<q}\MC{RC}(i) \\
= & \sum_{1\le k < \frac{q}{2}}k \cdot\sum_{1\le i\le q-k} d_u(v_i,v_{i+k}) +
\sum_{1\le k < \frac{q}{2}} (q-k) \sum_{1\le i\le k} d_u(v_i,v_{i+p-k}) +f(q) \\
= & \sum_{1\le k \le \frac{q}{2}} q\cdot \BIGLR{(}{)}{k\cdot\Delta(\MC{Q}) - \sum_{1\le i<k}(k-i)(\ell_i+\ell_{q-i})}
\end{align*}
For $1\le i<\frac{q}{2}$, the coefficient of $\ell_i$ and $\ell_{q-i}$ in the above summation is $q\cdot\sum_{i<k<\frac{q}{2}}(k-i)$, which equals $q\cdot\sum_{1\le k<\frac{q}{2}-i}k$ by substituting the variable $k$ by $k-i$. Therefore, we have
\begin{align*}
\sum_{1\le i<q}\MC{RC}(i) \ge \sum_{1\le k < \frac{q}{2}} q\cdot k\cdot \Delta(\MC{Q}) - \sum_{1\le k < \frac{q}{2}} q\cdot \sum_{1\le i< \frac{q}{2}-k}i\cdot(\ell_k+\ell_{q-k}).
\end{align*}
%
%
%
Since $\Delta(\MC{Q}) = \sum_{1\le i<q}\ell_i$, by further expanding $\Delta(\MC{Q})$, we obtain
$$\sum_{1 \le i<q}\MC{RC}(i) \ge \sum_{1\le k < \frac{q}{2}} q\cdot \sum_{\frac{q}{2} -k < i < \frac{q}{2}}i\cdot(\ell_k+\ell_{q-k}) + g(q).$$
\hfill$\Box$
\end{proof}



Now we are ready to prove Lemma~\ref{lemma_central_interaction_lower_bound}.

\begin{proof}[Proof of Lemma~\ref{lemma_central_interaction_lower_bound}]
We divide the total interaction to be lower-bounded, $\sum_{\frac{q}{4} \le i \le \frac{3}{4}q}\MC{RC}(i)$, into three parts which we discuss below.

\begin{figure}[h]
\centering
\vspace{-10pt}
\includegraphics[scale=0.8]{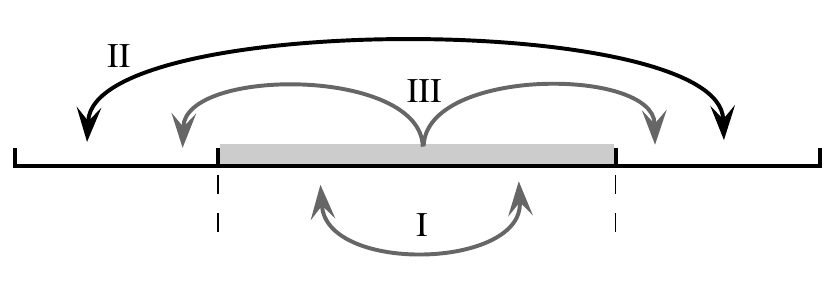}
\label{figure_central_region_lower_bound}
\vspace{-20pt}
\end{figure}

{\renewcommand\theenumi {\Roman{enumi}}
\begin{enumerate}
	\item
		the interaction between points from $\BIGLR{\{}{\}}{v_{\BIGLR{\lceil}{\rceil}{\frac{q}{4}}}, v_{\BIGLR{\lceil}{\rceil}{\frac{q}{4}}+1}, \ldots, v_{\BIGLR{\lfloor}{\rfloor}{\frac{3q}{4}}}}$.
		
		\smallskip
		
		The situation is equivalent to computing the overall interaction for a point set of $\frac{q}{2}$ points. By Lemma~\ref{lemma_counting_overall_routing_cost} with index replacement, the interaction is lower-bounded by $\sum_{1\le k < \frac{q}{4}} \frac{q}{2}\cdot \sum_{\frac{q}{4} -k < i < \frac{q}{4}}i\cdot(\ell_{\frac{q}{4}+k}+\ell_{\frac{3q}{4}-k}) + g^\prime(q),$ where
$g^\prime(q) = \frac{q}{2}\cdot\sum_{1\le i<\frac{q}{4}}i\cdot\ell_\frac{q}{2}$ if $\frac{q}{2}$ is even and $g^\prime(q) = 0$ otherwise.
		Dropping the items corresponding to $k < \frac{q}{12}$ from the first summation, we obtain $\frac{q}{2}\cdot\sum_{\frac{q}{6}q\le i \le \frac{q}{4}}i \cdot \sum_{\frac{q}{3}\le k\le \frac{2q}{3}}\ell_k$.
		
		\medskip
		
		For the remaining two cases, we consider the number of times each of the items from $\sum_{\frac{q}{3} \le k \le \frac{2q}{3}}\ell_k$ contributes to $\sum_{\frac{q}{4}\le i\le \frac{3q}{4}}\MC{RC}(i)$.
		
	\item
		the interaction between $\BIGLR{\{}{\}}{v_1, v_2, \ldots, v_{\BIGLR{\lceil}{\rceil}{\frac{q}{4}}}}$ and $\BIGLR{\{}{\}}{v_{\BIGLR{\lfloor}{\rfloor}{\frac{3q}{4}}}, v_{\BIGLR{\lfloor}{\rfloor}{\frac{3q}{4}}+1}, \ldots, v_q}$. 
		
		\smallskip
		
		For each $j,k$ such that $1\le j\le \frac{q}{4}$, $\frac{3q}{4} \le k < q$, the pair $d_u(v_j, v_k)$ contributes exactly once to the term $\MC{RC}(i)$ for each $i$ with $\frac{q}{4}\le i\le \frac{3q}{4}$. There are $\frac{1}{16}q^2$ such pairs, while there are $\frac{q}{2}$ different terms in the final summation $\sum_{\frac{q}{4} \le i \le \frac{3}{4}q}\MC{RC}(i)$. Therefore, we obtain a lower bound of $\frac{1}{32}q^3\cdot \sum_{\frac{q}{3}\le k\le \frac{2q}{3}}\ell_k$ for this part.
		
	\item
		the interaction between $\BIGLR{\{}{\}}{v_{\BIGLR{\lceil}{\rceil}{\frac{q}{4}}}, v_{\BIGLR{\lceil}{\rceil}{\frac{q}{4}}+1}, \ldots, v_{\BIGLR{\lfloor}{\rfloor}{\frac{3q}{4}}}}$ and other points.
		
		\smallskip
		
		For any specific interval $\ell_p$ with $\frac{q}{4} \le p \le \frac{3q}{4}$, we consider the number of pairs between $\BIGLR{\{}{\}}{v_{\BIGLR{\lceil}{\rceil}{\frac{q}{4}}}, v_{\BIGLR{\lceil}{\rceil}{\frac{q}{4}}+1}, \ldots, v_{\BIGLR{\lfloor}{\rfloor}{\frac{3q}{4}}}}$ and other points that contain this specific interval $\ell_p$. There are $p-\frac{q}{4}$ points, $\BIGLR{\{}{\}}{v_{\BIGLR{\lceil}{\rceil}{\frac{q}{4}}}, v_{\BIGLR{\lceil}{\rceil}{\frac{q}{4}}+1}, \ldots, v_{p}}$, which lie to the left of $v_p$ and form pairs with points from $\BIGLR{\{}{\}}{v_{\BIGLR{\lfloor}{\rfloor}{\frac{3q}{4}}}, v_{\BIGLR{\lfloor}{\rfloor}{\frac{3q}{4}}+1}, \ldots, v_q}$ that contain $\ell_p$. Similarly, the $\frac{3q}{4}-p$ points that lie to the right of $v_p$ also form pairs with points from $\BIGLR{\{}{\}}{v_1, v_2, \ldots, v_{\BIGLR{\lceil}{\rceil}{\frac{q}{4}}}}$ that contain $\ell_p$. Therefore there are $\frac{q}{4}\cdot \BIGP{p-\frac{q}{4}+\frac{3q}{4}-p} = \frac{q}{4}\cdot\frac{q}{2}$ such pairs. This is true for all $\MC{RC}(i)$ with $\frac{q}{4}\le i\le \frac{3q}{4}$. Therefore $\ell_p$ contributes $\frac{q}{4}\cdot\frac{q}{2}\cdot\frac{q}{2}$ times in the summation and we obtain a lower bound of $\frac{1}{16}q^3\cdot \sum_{\frac{q}{3}\le k\le \frac{2q}{3}}\ell_k$.
\end{enumerate}}

\noindent
Summing up the bounds we obtained in the three parts and we have this lemma.
\hfill$\Box$
\end{proof}

\begin{ap_lemma}{\ref{lemma_expectation_upper_bound}}
We have
{\small
$$\min\BIGLR{\{}{\}}{E\BIGLR{[}{]}{\frac{\beta\cdot(q-\beta)\cdot\Delta(\MC{Q})}{\MC{RC}(\beta)}}, \min_{1\le \gamma\le \frac{q}{3}}\BIGLR{\{}{\}}{\frac{\gamma\cdot(q-\gamma)\cdot\Delta(\MC{Q})}{\MC{RC}(\gamma)}, \frac{\gamma\cdot(q-\gamma)\cdot\Delta(\MC{Q})}{\MC{RC}(q-\gamma)}}} \le \frac{210}{59}.$$
}
\end{ap_lemma}

\begin{proof}[Proof of Lemma~\ref{lemma_expectation_upper_bound}]
This lemma holds trivially when $q \le 3$. For $q \ge 4$, by the definition of expected values, we have
\begin{align*}
& E\BIGLR{[}{]}{\frac{\beta\cdot(q-\beta)\cdot\Delta(\MC{Q})}{\MC{RC}(\beta)}}
= \sum_{\frac{q}{4} \le i \le \frac{3q}{4}} Pr\BIGLR{[}{]}{\beta = i} \cdot \frac{\beta\cdot(q-\beta)\cdot\Delta(\MC{Q})}{\MC{RC}(\beta)}
= \frac{\sum_{\frac{q}{4} \le i \le \frac{3q}{4}}i\cdot(q-i)\cdot\Delta(\MC{Q})}{\sum_{\frac{q}{4}\le i \le \frac{3q}{4}}\MC{RC}(i)}.
\end{align*}
First we have 
$$\sum_{\frac{q}{4} \le i \le \frac{3q}{4}}i(q-i)\cdot \Delta(\MC{Q}) = \BIGP{q\cdot\sum_{\frac{q}{4} \le i \le \frac{3q}{4}}i - \sum_{\frac{q}{4} \le i \le \frac{3q}{4}}i^2}\cdot\Delta(\MC{Q}) \le \frac{11}{96}q^3\Delta(\MC{Q}).$$
Depending on whether or not $\sum_{\frac{q}{3}\le k\le \frac{2q}{3}}\ell_i\ge \frac{11}{35}\Delta(\MC{Q})$, 
we distinguish between two cases.

\smallskip

If $\sum_{\frac{q}{3}\le k\le \frac{2q}{3}}\ell_i\ge \frac{11}{35}\Delta(\MC{Q})$,
then, by Lemma~\ref{lemma_central_interaction_lower_bound}, we have
$$\sum_{\frac{q}{4} \le i \le \frac{3q}{4}}\MC{RC}(i) \ge 
\sum_{\frac{q}{3}\le k\le \frac{2q}{3}}\ell_k \cdot \BIGP{\frac{3}{32}q^3 + \frac{q}{2} \cdot\sum_{\frac{q}{6}\le i \le \frac{q}{4}}i}
\ge \frac{11}{35}\Delta(\MC{Q}) \cdot \frac{59}{96\cdot 6} q^3,$$
$$\text{and} \quad E\BIGLR{[}{]}{\frac{\beta\cdot(q-\beta)\cdot\Delta(\MC{Q})}{\MC{RC}(\beta)}} 
\le \frac{11}{96}q^3\Delta(\MC{Q}) / \BIGP{\frac{11}{35}\Delta(\MC{Q}) \cdot \frac{59}{96\cdot 6} q^3} \le \frac{210}{59}.$$

\smallskip

On the other hand, 
if $\sum_{1 \le i \le \frac{q}{3}}(\ell_i+\ell_{q-i})\ge \frac{11}{35}\Delta(\MC{Q})$, then we have either $\sum_{1 \le i \le \frac{q}{3}}\ell_i\ge \frac{12}{35}\Delta(\MC{Q})$, or
$\sum_{1 \le i \le \frac{q}{3}}\ell_{q-i}\ge \frac{12}{35}\Delta(\MC{Q})$. 
Without loss of generality, assume that $\sum_{1 \le i \le \frac{q}{3}}\ell_i\ge \sum_{1 \le i \le \frac{q}{3}}\ell_{q-i} \ge \frac{12}{35}\Delta(\MC{Q})$.
%

\begin{figure}[h]
\centering
\includegraphics[scale=0.8]{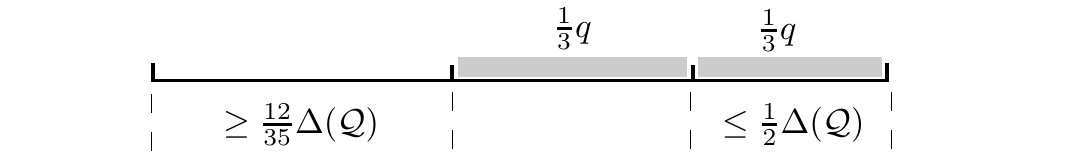}
\label{figure_boundary_lower_bound}
\end{figure}

In this case, we have $\sum_{1 \le i \le \frac{q}{3}}\ell_i + \sum_{\frac{q}{3} < i < \frac{2q}{3}}\ell_i \ge \sum_{\frac{2q}{3} \le i < q}\ell_i$. Therefore $\sum_{\frac{2q}{3} \le i < q}\ell_i \le \frac{\Delta(\MC{Q})}{2}$.
Let $p$ be the smallest integer such that $\ell_p > 0$. Counting the interaction between $\BIGLR{\{}{\}}{v_1, v_2, \ldots, v_p}$ and $\BIGLR{\{}{\}}{v_{p+1}, v_{p+2},\ldots, v_q}$, we have $\MC{RC}(p) \ge p\cdot \frac{q}{3} \cdot \frac{12}{35}\Delta(\MC{Q}) + p\cdot \frac{q}{3} \cdot \frac{1}{2}\Delta(\MC{Q})$.
Therefore,
\begin{align*}
\frac{p\cdot(q-p)\cdot\Delta(\MC{Q})}{\MC{RC}(p)} \le \frac{p\cdot q\cdot\Delta(\MC{Q})}{p\cdot q \cdot \Delta(\MC{Q}) \cdot \BIGP{\frac{1}{3} \cdot \frac{12}{35} + \frac{1}{3} \cdot \frac{1}{2}}} = \frac{210}{59}.
\end{align*}

The argument for the case $\sum_{1 \le i \le \frac{q}{3}}\ell_{q-i} \ge \sum_{1 \le i \le \frac{q}{3}}\ell_i$ is analogous. This proves the lemma.
\hfill$\Box$
\end{proof}




\subsection{Lower Bound}  \label{ap_subsection_lower_bounds}

Let $\MC{A}=\{a_1, a_2, \ldots, a_n\}$ be a set of numbers, where $a_i = i$ for all $1\le i\le n$,
and $(\MC{A},d)$ be the corresponding metric extracted from $\MC{A}$.
Let $(T,d_T)$ be an optimal ultra-metric embedding of $\MC{A}$ in terms of distance-weighted average stretch.
Without loss of generality, we can assume that $T$ is a binary tree.
Otherwise, we can always create dummy nodes to make $T$ binary without changing its sum of pairwise distances.
The following lemma characterizes the structure of $T$.

\medskip

\begin{ap_lemma}{\ref{lemma_ultrametric_lower_bound_partition}}
Let $T_L$ and $T_R$ be the left-subtree and the right-subtree of $T$ such that $a_1 \in T_L$.
Then, there exists an integer $k$, $1\le k<n$, such that $T_L$ is an ultra-metric containing $\{a_1, a_2, \ldots, a_k\}$ and $T_R$ is an ultra-metric containing $\MC{A} \backslash \{a_1, a_2, \ldots, a_k\}$.
\end{ap_lemma}

\begin{proof}[Proof of Lemma~\ref{lemma_ultrametric_lower_bound_partition}]
If not, let $\ell$ be the number of leaves in $T_L$, and denote by $\varphi$ the permutation on $\{1,2,\ldots,n\}$ such that
$T_L$ is an ultra-metric containing $\{a_{\varphi(1)}, a_{\varphi(2)}, \ldots, a_{\varphi(\ell)}\}$, where $a_{\varphi(1)} < a_{\varphi(2)} < \ldots < a_{\varphi(\ell)}$, and
$T_R$ is an ultra-metric containing $\{a_{\varphi(\ell+1)}, a_{\varphi(\ell+2)}, \ldots, a_{\varphi(n)}\}$, where $a_{\varphi(\ell+1)} < a_{\varphi(\ell+2)} < \ldots < a_{\varphi(n)}$.
Note that by our assumption, $a_{\varphi(\ell)} > a_{\varphi(\ell+1)}$.

Construct a new ultra-metric $\MC{T}_0$ as follows. 
The structure of $\MC{T}_0$ is identical to $T$. For each leaf node in $T$ that contains the singleton element, say $a_u$, we put the element $a_{\varphi^{-1}(u)}$ in the corresponding leaf node of $\MC{T}_0$.
The label of each internal node in $\MC{T}$ is set to be the diameter of the set of elements contained in the subtree rooted at it.

For each $i,j$ with $1 \le i < j \le \ell$ or $\ell < i < j\le n$,
since $i < j$ implies $a_{\varphi(i)} < a_{\varphi(j)}$ by the definition of $\varphi$,
we have $a_{\varphi(j)} - a_{\varphi(i)} \ge j - i$.
Therefore the label of each internal node in $\MC{T}_0$ is no larger than that of the corresponding internal node in $T$.
Furthermore, since $a_{\varphi(\ell)} > a_{\varphi(\ell+1)}$ by assumption, we have $a_\ell - a_1 < a_{\varphi(\ell)} - a_{\varphi(1)}$ and
$a_n - a_{\ell+1} < a_{\varphi(n)} - a_{\varphi(\ell+1)}$.
Therefore, the labels of the roots of the left-subtree and the right-subtree of $\MC{T}_0$ are strictly smaller than the labels
of their corresponding nodes in $T$.
Hence we can conclude that $\MC{R}({\MC{T}}) < \MC{R}(T)$, which is a contradiction to the optimality of $T$.
\hfill$\Box$
\end{proof}

\begin{ap_lemma}{\ref{lemma_cutting_lemma_lower_bound}}
Let $\delta_0$ be a constant such that our point set cutting lemma holds, then $\delta_0 \ge 2$.
\end{ap_lemma}

\begin{proof}[Proof of Lemma~\ref{lemma_cutting_lemma_lower_bound}]
Consider the set of numbers $\MC{A}$.
Assume that we cut $\MC{A}$ at a point $z \in (a_k, a_{k+1}]$, for some $1\le k<n$.
The left-hand side of the inequality in our cutting lemma is $k\cdot(n-k)\cdot (n-1)$,
while the right-hand side is $\sum_{1\le i \le k}\sum_{k < j\le n}(j-i) = \frac{1}{2}kn(n-k)$, where the equality follows from Equation~(\ref{ob_rc}) derived in \S~\ref{section_linear_time_cut}.
%
Therefore we have $$\delta_0 \ge \frac{k(n-k)(n-1)}{\frac{1}{2}kn(n-k)} = 2\cdot\frac{n-1}{n},$$
which converges to $2$ as $n$ tends to infinity.
Since this is true for all $k$ with $1\le k<n$, this lemma follows.
\hfill$\Box$
\end{proof}

\begin{ap_corollary}{\ref{cor_dominating_metric_lower_bound}}
Let $\MC{M} = (V,d)$ be a given metric and $\MC{D}(\MC{M})$ be the set of dominating tree metrics of $\MC{M}$. Then $$\inf_{(V^\prime,d^\prime) \in \MC{D}(\MC{M})}{\frac{\sum_{u,v \in V}d^\prime(u,v)}{\sum_{u,v \in V}d(u,v)}} \ge 2.$$
\end{ap_corollary}

\begin{proof}[Proof of Corollary~\ref{cor_dominating_metric_lower_bound}]
This corollary follows directly from Lemma~\ref{lemma_ultrametric_lower_bound_partition}, Lemma~\ref{lemma_cutting_lemma_lower_bound}, and induction on the size of $\MC{A}$.
\hfill$\Box$
\end{proof}



\subsection{Computing the Optimal Cut in Linear Time} \label{section_linear_time_cut}

In this section, we show how the best cut can be computed efficiently in linear time.
Let $\{a_1, a_2, \ldots, a_n\}$, $a_1 \le a_2 \le \ldots \le a_n$, be the given set points.
%
For each $k$ with $1\le k < n$, 
let $LS(k) = \sum_{1\le i< k}\BIGP{a_k-a_i}$ and $RS(k) = \sum_{k<i\le n}\BIGP{a_i-a_k}$ be the sum of the distances between $a_k$ and the points
to the left of $a_k$ and the sum of distances between $a_k$ and the points to the right of $a_k$, respectively. The first observation is that, for $i\le i<n$,
\begin{equation}
\MC{RC}(i) = (n-i)\cdot LS(i) + i\cdot RS(i). \label{ob_rc}
\end{equation}
The following lemma shows how these quantities can be computed recursively.

\begin{lemma} \label{lemma_routing_cost_computation_formula}
For $1\le k< n-1$, We have 
\begin{itemize}
	\item{$LS(k+1) = LS(k) + \sum_{1\le i\le k}\ell_k$, and}
	\item{$RS(k+1) = RS(k) - \sum_{k<i\le n}\ell_k$.}
\end{itemize}
\end{lemma}

\begin{proof}[Proof of Lemma~\ref{lemma_routing_cost_computation_formula}]
By definition, we have $LS(k+1) = \sum_{1\le i<k+1}\BIGP{\ell_k + a_k-a_i} = LS(k) + \sum_{1\le i\le k}\ell_k$,
and $RS(k+1) = \sum_{k+1<i\le n}\BIGP{a_i-a_k-\ell_k} = RS(k) - \sum_{k<i\le n}\ell_k$.
\hfill$\Box$
\end{proof}



By Lemma~\ref{lemma_routing_cost_computation_formula} and $(\ref{ob_rc})$,
we can compute in linear time the values $LS(k), RS(k)$, $\MC{RC}(k)$ for all $1\le k<n$, and the optimal cut. For any given interval $\MC{I} \subseteq [a_1,a_n]$, we can also compute the optimal cut inside $\MC{I}$ by the same approach.


\section{Approximating Euclidean Metrics by Their Spanning Trees}

\begin{figure*}[h]
\noindent\rule{\linewidth}{0.2mm}
\medskip
{{\sc Algorithm} {\em Euclidean-Spanning-Tree$\left(\MC{P}\right)$}} \newline
{Input: A set $\MC{P}$ of $n$ points in $\MC{R}^d$.} \newline
{Output: A pair $(\MC{T},r)$, which is a spanning tree $\MC{T}$ of $\MC{P}$ with root $r$.}

\begin{algorithmic}[1]
\IF{$\MC{P}$ is a singleton point set containing point $p$}
	\STATE Return $(\MC{P}, p)$.
\ENDIF
\STATE Let $\alpha = \frac{1}{4}$ be a constant.
\STATE Let 
$k$ be the index of dimension such that $\MC{L}_k(\MC{B}(\MC{P})) = \MC{L}_{max}(\MC{B}(\MC{P}))$.
\STATE Let $a_1 \le a_2 \le \ldots \le a_n$ be the coordinates of the projection of $\MC{P}$ into $k^{th}$ dimension, labelled in sorted order.
\STATE $p = \alpha\cdot(a_1+a_n)$, $q = (1-\alpha)\cdot(a_1+a_n)$.
\STATE $(\MC{P}_1, \MC{P}_2) \longleftarrow $ {\em 1d-cut}$\BIGP{\{a_1, a_2, \ldots, a_n\}, \BIGLR{[}{]}{p,q}}$.
\STATE $(T_1, r_1) \longleftarrow $ {\em Euclidean-Spanning-Tree}$(\MC{P}_1), (T_2, r_2) \longleftarrow $ {\em Euclidean-Spanning-Tree}$(\MC{P}_2)$.
\STATE Let $T \longleftarrow T_1 \cup T_2 \cup \{(r_1, r_2)\}$.
\STATE Return $\left(T, r_1\right)$.
\end{algorithmic}
\rule{\linewidth}{0.2mm} 
\caption{Algorithm for computing a spanning tree of low routing cost on Euclidean graphs.} \label{algorithm_spanning_tree_euclidean_graphs}
\end{figure*}

For convenience, let $\MC{F}$ be the collection of subsets of $\MC{P}$ which have occurred during the recursions. For any $\MC{Q} \in \MC{F}$, we denote by $\MC{T}[\MC{Q}]$ the subtree of $\MC{T}$ corresponding to $\MC{Q}$ and $e(\MC{Q})$ the edge connecting the two rooted subtrees corresponding to the two further partitions of $\MC{Q}$.
$e(\MC{Q})$ is defined to be a dummy self-loop with length zero if $\MC{Q}$ is a singleton set.
The following lemma provides an upper-bound on the pairwise distances.

\begin{lemma} \label{lemma_euclidean_radius_bound}
For any $p,q \in \MC{P}$, we have $d_\MC{T}(p,q) \le \frac{2}{\alpha}d\sqrt{d}\cdot\MC{L}_{max}(\MC{B}(\MC{P}))$.
\end{lemma}

\begin{proof}[Proof of Lemma~\ref{lemma_euclidean_radius_bound}]
Let $A_1 \supset A_2 \supset \ldots \supset A_a$, $A_i \in \MC{F}$ for $1\le i\le a$, be the subsets of $\MC{P}$ occurred during the recursions to which $p$ belongs,
and $B_1 \supset B_2 \supset \ldots \supset B_b$, $B_j \in \MC{F}$ for $1\le j\le b$, be the subsets to which $q$ belongs.
Note that $A_1 = B_1 = \MC{P}$, $A_a = \{p\}$, and $B_b = \{q\}$.
%
From the construction of $\MC{T}$, we have 
\begin{align*}
d_{\MC{T}}(p,q) & \le d_{\MC{T}[A_1]}(p,r_1) + \BIGC{e(\MC{P})} + d_{\MC{T}[B_1]}(r_2,q) \le \sum_{1\le i\le a}\BIGC{e(A_i)} + \BIGC{e(\MC{P})} + \sum_{1\le j\le b}\BIGC{e(B_j)},
\end{align*}
where $r_1$ and $r_2$ are the roots of $\MC{T}[A_1]$ and $\MC{T}[B_1]$.
Since the longest straight-line distance inside a hyper-rectangle is bounded by its longest diagonal,
we have $\BIGC{e(Q)} \le \sqrt{d}\MC{L}_{max}(\MC{B}(Q))$ for any subset $Q \in \MC{F}$.
Furthermore, since we always cut along the longest side of the bounding box,
we have $\MC{L}_{max}(\MC{B}(A_{i+d})) \le (1-\alpha)\MC{L}_{max}(\MC{B}(A_i))$ and $\MC{L}_{max}(\MC{B}(B_{j+d})) \le (1-\alpha)\MC{L}_{max}(\MC{B}(B_j))$
for all $1\le i \le a-d$ and $1\le j\le b-d$.
Therefore, it follows that 
\begin{align*}
d_\MC{T}(p,q) & \le \sum_{1\le i\le a}\sqrt{d}\MC{L}_{max}(\MC{B}(A_i)) + \sqrt{d}\MC{L}_{max}(\MC{B}(\MC{P})) + \sum_{1\le j\le b}\sqrt{d}\MC{L}_{max}(\MC{B}(B_j)) \\
& \le 2d\cdot\sum_{i \ge 1}\sqrt{d}(1-\alpha)^i\MC{L}_{max}(\MC{B}(\MC{P})) + \sqrt{d}\MC{L}_{max}(\MC{B}(\MC{P})) \\
& \le \frac{2}{\alpha}d\sqrt{d}\cdot\MC{L}_{max}(\MC{B}(\MC{P})),
\end{align*}
where in the second last inequality we collect every $d$ items from the summation of the first inequality and then combine them together into a geometric series.
\hfill$\Box$
\end{proof}


\begin{ap_lemma}{\ref{lemma_1_d_weighted_cutting}}
Given a set of real numbers $A = \left\{a_1, a_2, \ldots, a_n\right\}$, $a_1\le a_2\le \ldots \le a_n$
and an interval $\MC{I} = [\ell,r]$ such that $\MC{I} \subseteq [a_1,a_n]$, there exists a cutting point $z \in \MC{I}$ such that the following holds.
$$L_A(z) \cdot \BIGP{n-L_A(z)} \cdot \BIGC{\MC{I}} \le \delta_0\cdot \sum_{1\le i\le L_A(z)}\enskip \sum_{L_A(z)<j\le n}(a_j-a_i),$$
where $L_A(z) = \BIGC{\{a\in A: a < z\}}$ is the number of elements in A that are smaller than $z$
and $\delta_0 \le \frac{210}{59}$ is a constant.
\end{ap_lemma}

\begin{proof}[Proof of Lemma~\ref{lemma_1_d_weighted_cutting}]
We say that an interval degenerates if it has length zero.
First we argue that, if there are degenerating intervals at $a_1$, then it is always worse to cut at those degenerating intervals.
Let $k$, $1\le k \le n$, be the largest index such that $a_1 = a_2 = \ldots = a_k$.
Observe that, for any $i,j$ with $1\le i,j \le k$, we have $\MC{RC}(i) = \frac{i}{j} \cdot \MC{RC}(j)$. On the other hand, for $1\le i<k$ and $1\le j\le k-i$, we have
$$\frac{(i+j)(n-i-j)}{i(n-i)} = \frac{i(n-i)+j(n-2i-j)}{i(n-i)} \le \frac{i+j}{i} = \frac{\MC{RC}(i+j)}{\MC{RC}(i)},$$
which implies that $\frac{(i+j)(n-i-j)}{\MC{RC}(i+j)} \le \frac{i(n-i)}{\MC{RC}(i)}$ and therefore cutting at $(a_k,a_{k+1}]$ is always better than cutting at degenerating intervals at $a_1$.
Similarly, we can argue that, it is always worse to cut at the degenerating intervals at $a_n$, if there is any.

Now we argue that there will be a feasible cut satisfying the criterion.
According to the given interval $\MC{I} = [a,b]$ and the point set $A$, we create a new point set $B = \{b_1, b_2, \ldots, b_n\}$ as follows.
\vspace{-10pt}
\begin{align*}
\hspace{0.2\textwidth}\text{For $1\le i\le n$,} \quad b_i = \begin{cases}
\ell & \text{if $a_i < \ell$,} \\
a_i & \text{if $\ell \le a_i \le r$,} \\
r & \text{otherwise.}
\end{cases}
\end{align*}

Let $z$ be the best cut of $B$ in $\MC{I}$. By the above argument, we have $\ell < z < r$ and therefore $L_A(z)=L_B(z)$.
By
Lemma~\ref{lemma_1_d_cutting}, we have
$L_B(z) \cdot \BIGP{n-L_B(z)} \cdot \left|\MC{I}\right| \le \frac{210}{59} \sum_{b_i < z \le b_j}(b_j-b_i)$.
According to our setting, we have $(b_j-b_i) \le (a_j-a_i)$ for all $1\le i<j\le n$. Therefore $L_A(z) \cdot \BIGP{n-L_A(z)}\cdot \BIGC{\MC{I}} \le \frac{210}{59} \sum_{1\le i\le L_A(z)} \sum_{L_A(z) < j\le n}(a_j-a_i)$
as claimed.
\hfill$\Box$
\end{proof}


\begin{ap_theorem}{\ref{thm_euclidean_tree}}
Given a set of points $\MC{P}$ in $\MC{R}^d$, Algorithm \emph{Euclidean-Spanning-Tree} computes a spanning tree $\MC{T}$ of $\MC{P}$ such that the distance-weighted
average stretch of $\MC{T}$ with respect to $\MC{P}$ is at most $16\delta_0\cdot d\sqrt{d}$, where $\delta_0 \le \frac{210}{59}$ is the constant in our point set cutting lemma.
\end{ap_theorem}

\begin{proof}[Proof of Theorem~\ref{thm_euclidean_tree}]
If $\left|\MC{P}\right| = 1$, then this theorem holds trivially. Otherwise,
by Lemma~\ref{lemma_euclidean_radius_bound}, Lemma~\ref{lemma_1_d_weighted_cutting}, and the fact that the length of the restricted interval is $(1-2\alpha)\cdot\MC{L}_{max}(\MC{B}(\MC{P}))$, 
we have 
$$\MC{R}_\MC{T}(\MC{P}_1, \MC{P}_2) \le \BIGC{\MC{P}_1}\cdot\BIGC{\MC{P}_2}\cdot\frac{2}{\alpha}d\sqrt{d}\cdot\MC{L}_{max}(\MC{B}(\MC{P})) \le \frac{2\delta_0}{\alpha(1-2\alpha)}d\sqrt{d}\MC{R}(\MC{P}_1,\MC{P}_2).$$
This holds for all recursions. 
Choose $\alpha$ to be $\frac{1}{4}$ and this theorem follows directly by induction on the depth of recursion.
\hfill$\Box$
\end{proof}

\end{appendix}

\end{document}